\documentclass[a4paper]{article}
\pdfoutput=1 

\usepackage{amsmath}
\usepackage{a4wide}
\usepackage{color}
\usepackage{amssymb}
\usepackage{amsthm}
\usepackage{graphicx}
\usepackage{wrapfig}
\usepackage{algorithmic}
\usepackage{algorithm}
\usepackage{hyperref}
\usepackage{multirow}
\usepackage{verbatim}
\usepackage{cite}
\usepackage{lineno}

\newtheorem{lemma}{Lemma}
\newtheorem{theorem}{Theorem}

\newtheorem{corollary}{Corollary}

\newtheorem{observation}{Observation}

\newcommand{\Poly}{\ensuremath{\mathcal{P}}}               
\newcommand{\bd}{\ensuremath{\partial \Poly} }                 
\newcommand{\Vis}{\ensuremath{\mathrm{Vis}_\Poly}} 
\newcommand{\VisC}{\ensuremath{\mathrm{Vis}_{\mathcal{C}}}} 
\newcommand{\Rin}{\ensuremath{r}}      
\newcommand{\Rout}{\ensuremath{\bar{r}}}      
\newcommand{\Hout}{\ensuremath{\bar{h}}}      
\newcommand{\chain}{\ensuremath{{\mathrm{Chain}}}}     
\newcommand{\region}{\ensuremath{{{\mathcal R}}}}     
\newcommand{\regionC}{\ensuremath{{{\mathcal R_{\mathcal C}}}}}     
\newcommand{\regone}{\ensuremath{{{\mathcal R_{\mathcal C_{1}}}}}}     
\newcommand{\regtwo}{\ensuremath{{{\mathcal R_{\mathcal C_{2}}}}}}
\newcommand{\coneC}{\ensuremath{{{\Delta_{\mathcal C}}}}} 
\newcommand{\C}{\ensuremath{{\mathcal C}}}

\newcommand{\marrow}{\marginpar[\hfill$\longrightarrow$]{$\longleftarrow$}}
\newcommand{\remark}[2]{\textcolor{red}{\textsc{#1 says:} \marrow\textsf{#2}}}
\newcommand{\rodrigo}[1]{\remark{Rodrigo}{#1}}

\graphicspath{{figures/}}

\title{Computing a visibility polygon using few variables\thanks{A preliminary version of this paper appeared in the proceedings of the 22nd International Symposium on Algorithms and Computation (ISAAC 2011)~\cite{bkls-cvpufv-11}.}}

\author{
Luis Barba
\thanks{Universit\'e Libre de Bruxelles (ULB), Brussels, Belgium. {\tt \{lbarbafl,stefan.langerman\}@ulb.ac.be}}
\and Matias Korman
\thanks{Universitat Polit\`{e}cnica de Catalunya (UPC), Barcelona, Spain. {\tt  \{matias.korman, rodrigo.silveira\}@upc.edu}. With the support of the Secretary for Universities and Research of the Ministry of Economy and Knowledge of the Government of Catalonia, the European Union, the FP7 Marie Curie Actions Individual Fellowship PIEF-GA-2009-251235, and ESF EUROCORES programme EuroGIGA - ComPoSe IP04 - MICINN Project EUI-EURC-2011-4306.}
\and Stefan Langerman\footnotemark[2]
\thanks{Directeur de Recherches du FRS-FNRS.}
\and Rodrigo I. Silveira\footnotemark[3]
}

\date{}

\begin{document}
\maketitle

\begin{abstract}
We present several algorithms for computing the visibility polygon of a simple polygon $\Poly$ of $n$ vertices (out of which $r$ are reflex) from a viewpoint inside $\Poly$, when $\Poly$ resides in read-only memory and only few working variables can be used.
 The first algorithm uses a constant number of variables, and outputs the vertices of the visibility polygon in $O(n\Rout)$ time, where $\Rout$ denotes the number of reflex vertices of $\Poly$ that are part of the output. Whenever we are allowed to use $O(s)$ variables, the running time decreases to $O(\frac{nr}{2^{s}}+n\log^2 r)$ (or $O(\frac{nr}{2^{s}}+n\log r)$ randomized expected time), where $s\in O(\log r)$. This is the first algorithm in which an exponential space-time trade-off for a geometric problem is obtained.
\end{abstract}

\section{Introduction}

The \emph{visibility polygon} of a simple polygon $\Poly$ from a viewpoint $q$ is the set of all points of $\Poly$ that can be seen from $q$, where two points $p$ and $q$ can see each other whenever the segment $pq$ is contained in $\Poly$.
The visibility polygon is a fundamental concept in computational geometry and one of the first problems studied in planar visibility. 
The first correct and optimal algorithm for computing the visibility polygon from a point was found by  Joe and Simpson\cite{js-clvpa-87}. It computes the visibility polygon from a point in linear time and space. 
We refer the reader to the survey of O'Rourke \cite{r-v-04} and the book of Gosh~\cite{g-vap-07} for an extensive discussion of such problems.



In this paper we look for an algorithm that computes the visibility polygon of a given point and uses few variables. This kind of algorithm not only provides an interesting trade-off between running time and memory needed, but is also useful in portable devices where important hardware constraints are present (such as the ones found in digital cameras or mobile phones). 
In addition, this model has direct applications in concurrent environments where several devices with limited memory resources perform some computation on a large centralized input. 
Since several devices may access the input at the same time, allowing writing to the input memory can result in compromising its integrity.

A significant amount of research has focused on the design of algorithms that use few variables, some of them even dating from the 80s \cite{mp-ssls-80}. Although many models exist, most of the research considers that the input is in some kind of read-only data structure. In addition to the input values, we are allowed to use few additional variables (typically a variable holds a logarithmic number of bits). 

One of the most studied problems in this setting is that of selection. For any constant $\epsilon \in (0,1)$, Munro and Raman \cite{mr-sromswmdm-96} gave an algorithm that runs in $O(n^{1+\epsilon})$ time and uses $O(1/\epsilon)$ variables. Frederickson~\cite{Frederickson87} extended this result to the case in which $s$ working variables are available (and $s\in \Omega (\log n)\cap O(2^{\log n/\log^* n})$). Raman and Ramnath~\cite{rr-iubtstsls-98} gave several exact and approximation algorithms for the case in which fewer variables are available. Among other results, they provide a $2/3$-approximation of the median that runs in $O(sn^{1+1/s})$ time using $O(s)$ variables (for $s\in o(\log n)$), or $O(n\log n)$ time, using $O(\log n)$ variables. More recently Chan \cite{Chan} provided several lower bounds for performing selection with few variables. 

In recent years there has been a growing interest in finding algorithms for geometric problems that use a constant number of variables.
An early example is the well-known gift-wrapping algorithm (also known as Jarvis march~\cite{j-ich-73}), which can be used to  report the points on the convex hull of a set of $n$ points in $O(n\Hout)$ time using a constant number of variables, where \Hout\ is the number of vertices on the convex hull. 
Recently, Asano and Rote~\cite{ar-cwagp-09} and afterwards Asano {\em et al.} \cite{amw-cwaspsp-10,abbkmrs-mcasp-11} gave efficient methods for computing well-known geometric structures, such as the Delaunay triangulation, the Voronoi diagram, a polygon triangulation, and a minimum spanning tree (MST) using a constant number of variables. These algorithms run in $O(n^2)$ time (except computing the MST, which needs $O(n^3)$ time). Observe that, since these structures have linear size, they are not stored but reported. 
 Prior to this work, there was no algorithm for computing the visibility polygon in memory-constrained models. Indeed, this problem was explicitly posed as an open problem by Asano {\em et al.}~\cite{amrw-cwagp-10} for the case in which only a constant number of variables are allowed.

\subsubsection*{Results}

In this paper we present a novel approach for computing the visibility polygon of a given point inside a simple polygon. 
It is easy to see that reflex vertices have a much larger influence on the visibility polygon than convex vertices. Therefore, whenever possible we express the running time of our algorithms not only in terms of $n$, the complexity of $\Poly$, but also in terms of  $\Rin$ and $\Rout$ (the number of reflex vertices of $\Poly$ that are present in the input and in the output, respectively). This approach continues a line of research relating the combinatorial and computational properties of polygons to the number of their reflex vertices. We refer the reader to \cite{bdhilm-ghsc-07,akpv-got-12,bchm-gwp-10} and references found therein for a deep review of existing similar results.

In Section~\ref{sec:Preliminaries} we begin the paper with some preliminaries, followed by some observations and basic algorithms in Section~\ref{sec:SimpleAlgorithm}. In Section \ref{section:SequentialAlgorithms} we give an output-sensitive algorithm that reports the vertices of the visibility polygon in $O(n\Rout)$ time using $O(1)$ variables. 
 Using this algorithm as a stepping stone, in Section \ref{Section:Divide and conquer} we present a divide-and-conquer algorithm. This algorithm runs in $O(\frac{nr}{2^{s}}+n\log^2 r)$ time (or $O(\frac{nr}{2^{s}}+n\log r)$ randomized expected time)  using $O(s)$ variables (for any $s\in O(\log r)$), giving an exponential trade-off between running time and space. 
 

{\bf Remark:} prior to this research there was no known method for computing visibility polygons using few variables. Following the conference version of this paper~\cite{bkls-cvpufv-11}, De {\em et al.}~\cite{dmn-seavpsp-12} provided a linear-time algorithm that uses $O(\sqrt{n})$-variables. Parallel to this research, Barba {\em et al.}~\cite{bklss-sttosba-13} gave a general method for transforming stack-based algorithms into memory constrained workspaces. Since Joe and Simpson's algorithm for computing the visibility polygon~\cite{js-clvpa-87} is stack-based, their approach can be used for this problem as well. 

\section{Preliminaries}
\label{sec:Preliminaries}
\subsubsection*{Model definition and considerations on input/output precision}
We use a slight variation of the constant workspace model, introduced by Asano and Rote~\cite{ar-cwagp-09}. In this model the input of the problem resides in a read-only data structure and we are allowed to perform random access to any of the values of the input in constant time. 
An algorithm can use a constant number of variables and we assume that each variable or pointer contains a data word of $O(\log n)$ bits. 
Implicit storage consumption required by recursive calls is also considered part of the workspace. This model is also referred as {\em log space}~\cite{AB09} in the complexity literature.

Many other similar models have been studied. 
We note that in some of them (like the {\em streaming}~\cite{gk-seocqs-01} or the {\em multi-pass} model~\cite{cc-mpga-07}) the values of the input can only be read once or a fixed number of times. 
As in the constant workspace model of  Asano and Rote~\cite{ar-cwagp-09}, our model allows scanning the input as many times as necessary. 
However, our model differs from theirs in two aspects: we are allowed to use a workspace of $O(s)$ variables (instead of $O(1)$), and we do not require random access to the vertices of the input. 

The input to our problem is a simple polygon $\Poly$ in a read-only data structure and a point $q$ in the plane, from where the visibility polygon needs to be computed.
We do not make any assumptions on whether the input coordinates are rational or real numbers (in some implicit form). The only operations that we perform on the input are determining whether a given point is above, below or on a given line and determining the intersection point between two lines. In both cases, the line is defined as passing through two points of the input, hence both operations can be expressed as finding the root of linear equations whose coefficients are values of the input. We assume that these two operations can be done in constant time. Moreover, if the coordinates of the input are algebraic values, we can express the coordinates of the output as ``the intersection point of the line passing through points $p_i$ and $p_j$ and the line passing through $p_k$ and $p_l$" (where $p_i,p_j,p_k$ and $p_l$ are vertices of the input).

\subsubsection*{Definitions and basic properties}
We say that a point $p$ is \emph{visible} from $q$ (with respect to $\Poly$) if and only if the segment $pq$ is contained in $\Poly$ (note that we regard $\Poly$ as a closed subset of the plane). The set of points visible from $q$ is called the \emph{visibility polygon} of \Poly, and is denoted \Vis. Note that if $q$ is outside of \Poly\ then \Vis\ is by definition empty. Thus, when considering visibility with respect to polygons, we always assume that $q$ is inside the polygon. From now on, we assume that $q$ is fixed, hence we omit the ``with respect to $q$" term.

We assume we are given $\Poly$ as a list of its vertices in counterclockwise order along its boundary (denoted by $\bd$). Let $p_0$ be a point on $\bd$ closest to $q$ on the horizontal line passing through $q$. It is easy to see that $p_0$ is visible and can be computed in linear time. In the following, we will treat $p_0$ as a vertex of \Poly\ (even though it does not need to be one). By implicitly reordering the vertices of the input, we can assume that we are given the vertices of $\Poly$ in counterclockwise direction starting from $p_0$ (i.e., $\Poly = (p_0, \ldots, p_n)$). 

For simplicity of exposition, we assume that the vertices of $\Poly$ are in a weak general position; that is, we assume that there do not exist two vertices $p,p'\in\Poly$ such that $p, p'$, and $q$ are aligned (but we note that the algorithms can be extended easily for the general case). 

Along this paper, we will often work with polygonal chains (instead of polygons). However, we will restrict our scope to polygonal chains contained in $\bd$. For any two points $a,b$ on $\bd$, there is a unique path from $a$ to $b$ that travels counterclockwise along $\bd$; let $\C = \chain(a,b)$ be the set of points traversed in this path, including both $a$ and $b$ (this set is called the {\em chain} between $a$ and $b$). We say that $a$ and $b$ are the \emph{endpoints} of $\mathcal C$, and we refer to the rest of the points on $\mathcal C$ as its \emph{interior points}. 

We now extend the concept of visibility to chains. Due to technical reasons, we define this concept for chains contained in $\bd$ whose endpoints are visible from $q$.  We say that a chain $\C$ is \emph{independent} if and only if its endpoints are visible points of $\Poly$. Given a chain $\C=\chain(a,b)$ with endpoints $a$ and $b$, let \regionC\ denote the polygon enclosed by the union of \C\ and the segments $qa$ and $qb$ (equivalently, we use the notation $\region(a,b)$).

Given a point $q\in \mathbb{R}^2$ and an independent chain \C\ with endpoints $a$ and $b$, we say that a point $p\in \mathbb{R}^2$ is \emph{visible} from $q$ with respect to \C\ if and only if $p$ is visible from $q$ with respect to \regionC.
 The set of points that are visible with respect to $\C$ is called the \emph{visibility polygon} of $\C$, and is denoted \VisC. We start by observing that both concepts of visibility are equivalent (hence we need not distinguish between them).

\begin{observation}
If $\mathcal C = \chain(a,b)$ is an independent subchain, then \regionC\ is a simple polygon. Moreover, a point $x\in \C$ is visible with respect to $\Poly$ if and only if it is visible with respect to $\C$. 
\end{observation}

The above observation certifies that visibility within independent chains is well-defined. Since we will only consider independent chains, from now on we omit the ``with respect to $\Poly$" (or to $\C$), and simply say that a point $p$ is visible. 

True to its name, the visibility polygon of an independent chain $\C \subseteq \bd$ can be computed without having knowledge of the remainder of $\Poly$. 

\begin{observation}\label{DivideObs}
Let $\mathcal C = \chain(a,b)$ be an independent chain such that $\C = (a, p_i, \ldots, p_j, b)$ for some $0<i<j<n$. Let $x$ be a visible interior point of $\C$ lying on the edge $p_k p_{k+1}$ for some $i\leq k < j$. If we let $\mathcal C_1 = \{a, p_i, \ldots, p_k, x\}$ and $\mathcal C_2 = \{x, p_{k+1}, \ldots, p_j,b\}$,
then $$\VisC = \mathrm{Vis}_{\mathcal C_1}(q) \cup \mathrm{Vis}_{\mathcal C_2}(q),$$ 
$$\mathrm{Vis}_{\mathcal C_1}(q) \cap \mathrm{Vis}_{\mathcal C_2}(q) = qx.$$
\end{observation}



Given a point $p$ on the plane, let $\rho_q(p)$ be the ray emanating from $q$ and passing through $p$.
We define $\theta_q(p)$ as the angle that $\rho_q(p)$ makes with the positive $x$-axis, $0\leq \theta_q(p) < 2\pi$. 
We call $\theta_q(p)$ the \emph{CCW-angle} of $p$.




\begin{figure}[tb!]
\centering
\includegraphics{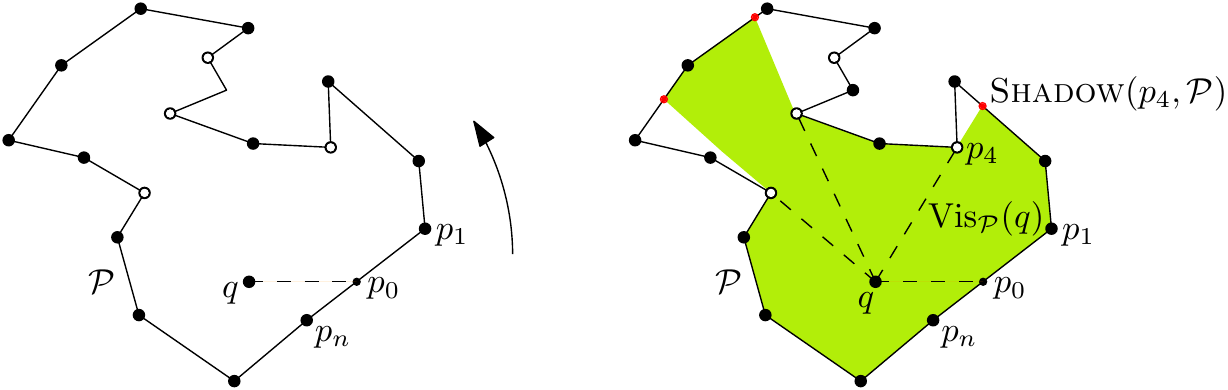}
\caption{Left:	general setting, vertices that are reflex with respect to $q$ are shown with a white point (black otherwise). Right: the visibility polygon $\Vis$.
}
\label{fig:polygon}
\end{figure}

We also need to define what a \emph{reflex} vertex is in our context. Given any vertex $p_k$, the line $\ell_k$ passing through $p_k$ and $q$ splits $\Poly \setminus \ell_k$ into disjoint components. A vertex $p_k$ is \emph{reflex with respect to $q$} if the angle at the vertex interior to $\Poly$ is more than $\pi$ and the vertices $p_{k-1}$ and $p_{k+1}$ lie on the same connected component of $\mathbb{R}^2\setminus \ell_k$ (see Fig. \ref{fig:polygon}). 
Observe that any vertex that is reflex with respect to $q$ is a reflex vertex (in the usual sense), but the converse is not true. Since the point $q$ is fixed, from now on we omit the ``with respect to $q$" term and simply refer to these points as reflex. 
Note that being reflex is a local property that can be verified in $O(1)$ time.
Intuitively speaking, reflex vertices with respect to $q$ are the vertices where important changes occur in the visibility polygon. That is, where the polygon boundary can change between visible or not-visible. Let $\Rin$ be the number of reflex of  vertices of $\Poly$. We also define $\Rout$ as the number of reflex vertices of $\Poly$ that are present in $\Vis$. Naturally, we always have $\Rout\leq \Rin < n$.

Given two points $p$ and $p'$ on a chain $\C=\chain(a,b)$, we say that $p$ lies \emph{before} $p'$ (resp. $p'$ lies \emph{after} $p$) if, when we walk from $a$ towards $b$ along $\C$, we first pass through $p$ and then through $p'$. We say that a chain is {\em visible} if all the points of the chain are visible. A visible chain $\C=\chain(a,b)$ is \emph{CCW-maximal} if no other visible chain starting at $a$ strictly contains ${\mathcal C}$. In this case, we say that $\C$ is the maximal chain {\em starting} at $a$ and {\em ending} at $b$.


Given a visible reflex vertex $v$ on a chain $\C$, we say that a point $w\neq v$ on $\bd$ is the \emph{shadow} of $v$ if $w$ is collinear with $q$ and $v$, and $w$ is visible from $q$ (this point is denoted by $\textsc{Shadow}$($v,\C$)). Due to the general position assumption, $w$ is uniquely defined and  must be an interior point of an edge. That is, each visible reflex vertex is uniquely associated to a shadow point (and vice versa). We say that a visible reflex vertex $v$ is of \emph{type R} (resp. \emph{type L}) if its shadow lies after (resp. before) $v$; see Fig.~\ref{fig:RightLeftReflex and ConeDefs} (left). 
Equivalently, a vertex $v$ is of type R if $q,v,x$ and $q,v,y$ make a \emph{right} turn (where $x$ and $y$ are the predecessor and the successor of $v$ on $\C$, respectively). 
Analogously, vertices of type L make a \emph{left} turns instead. 

A \emph{ray shooting query} is a basic operation that, given an independent chain \C\ and a point $x\in \C$, considers the ray $\rho_q(x)$ and reports the last visible point in $\rho_q(x)$ with respect to \C\ (i.e., the one furthest away from $q$). Observe that when $x$ is a visible reflex vertex we obtain its shadow. We denote the output of this operation by \textsc{RayShooting}$(x, \mathcal{C})$. It is easy to see that a ray shooting query can be performed in linear time, using only $O(1)$ extra variables, by scanning the edges of \C\ one by one and computing their intersections with $\rho_q(x)$. 

Finally, for $\C = \chain(a,b)$  we define $\coneC$ as the cone with apex $q$ that contains every point in the plane with CCW-angle between $\theta_q(a)$ and $\theta_q(b)$; see Fig.~\ref{fig:RightLeftReflex and ConeDefs} (right).


\section{Understanding the visibility polygon}
\label{sec:SimpleAlgorithm}
The basic scheme of our algorithms is to partition the input polygon into independent subchains so that the visibility within one subchain is unaffected by the others. We will use $\bd$ as the starting chain, thus the first chain will be closed, but the following chains will be open. Notice that since \Poly\ is simple, any chain of it will be simple too. 

In this section we present some observations about the independence between chains.

\begin{figure}[tb]
\centering
\includegraphics{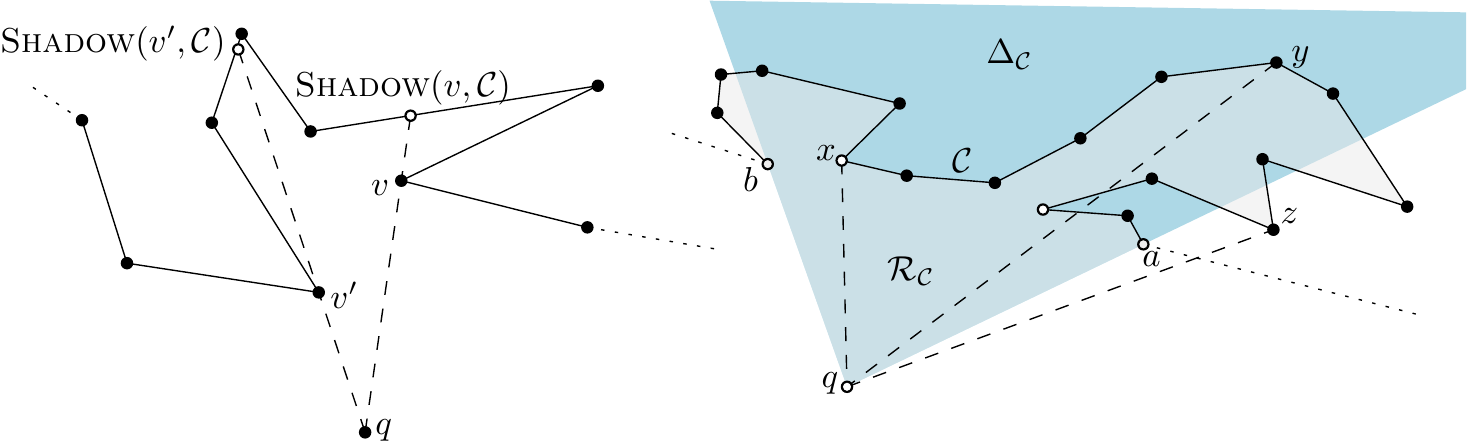}
\caption{\small Left: $v$ is a reflex vertex of type R, while $v'$ is of type L. The chain between $v$ and $\textsc{Shadow}(v,\mathcal C)$, together with the segment joining them, bound a simple polygon containing no visible points other than $v$ and $\textsc{Shadow}(v,\mathcal C)$.  Observe that the chain $\mathcal{C}(\textsc{Shadow}(v,\mathcal C),\textsc{Shadow}(v',\mathcal C))$ is CCW-maximal. 
Right: A polygonal chain $\mathcal C = \chain(a,b)$ and its associated polygon $\region(a,b) = \regionC$. Point $x$ is visible while points $y$ and $z$ are not. 
Note that every visible point of $\C$ lies inside $\coneC$.}
\label{fig:RightLeftReflex and ConeDefs}
\end{figure}

\begin{observation}\label{obs:VisibleChains}
Let $v$ be a visible reflex vertex of $\Poly$ of type R (resp. L) whose shadow is $w$. The $\chain(v, w)$ (resp. $\chain(w, v)$) contains no visible point other than its endpoints. In particular, a visible chain cannot contain an interior reflex vertex.
\end{observation}


The following important lemma characterizes the endpoints of CCW-maximal chains.

\begin{lemma}\label{lem_key}
Let $\C =\chain(a,b)$ be an independent chain, let $p\in \C$ be a visible point and let $v$ the first visible reflex vertex encountered when walking from $p$ towards $b$ (or $b$ if none exists). Let $\chain(p,p')$ be the CCW-maximal chain starting at $p$. The point $p'$ is either equal to $v$ (if $v=b$ or $v$ is of type R), or equal to the shadow of $v$ (if $v$ is of type L).
\end{lemma}
\begin{proof}
Clearly, all the points lying after $p$ are visible if and only if $p'=b$.
If $p' \neq b$, then, when walking on $\C$, $p'$ is the last visible point of the chain before a change in visibility occurs (i.e. points at distance $\varepsilon >0$ lying after $p'$ are not visible for sufficiently small values of $\varepsilon$).
This can happen for only two reasons: either $p'$ is a reflex vertex of type R, or there is some reflex vertex $v'$ of type L such that $p'$ is the shadow of $v'$.
In the former case, $p'$ is equal to $v$ since no reflex vertex lies in the interior of $\chain(p,p')$ by Observation~\ref{obs:VisibleChains}.
In the latter case, $v'$ must be the first visible reflex vertex lying after $p$ since no point between $p'$ and $v'$ can be visible by Observation~\ref{obs:VisibleChains}.
\end{proof}

The following result is a direct consequence of Lemma~\ref{lem_key} that allows us to report $\VisC$ of an independent chain $\C$ with no interior reflex vertices. 

\begin{corollary}\label{corollary:CharacterizationVisibleReflex}
Let $p \in \bd$ be a visible point such that $p$ is not a reflex vertex of type R. 
Let $v$ be the first visible reflex vertex lying after $p$ and let $w$ be its shadow. The following statements hold:
\begin{itemize}
\vspace{-1mm}\item[-] If $v$ is of type R, then $Chain(p, v)$ is CCW-maximal.
\vspace{-1mm}\item[-] If $v$ is of type L, then $Chain(p, w)$ is CCW-maximal.
\end{itemize}
\end{corollary}

\section{Output Sensitive Algorithm}\label{section:SequentialAlgorithms}
In this section we present an algorithm that will be used as stepping stone for our divide and conquer algorithm presented in  Section~\ref{Section:Divide and conquer}.
This base algorithm computes the visibility polygon of an independent chain using $O(1)$ extra variables.

For this purpose, we introduce an operation that we call $\textsc{NextVisReflex}$. 
This operation receives an independent chain $\C=\chain(a,b)$ and a visible point $p$ on $\C$. 
Its objective is to compute the next \emph{visible} reflex vertex lying after $p$ on \C, along with its shadow. That is, let $v$ be the first reflex vertex lying after $p$ on $\C$. If $v$ is visible, then $\textsc{NextVisReflex}(p,\C)$ should return $v$ and its shadow. Otherwise, we know that at some point when walking along the path from $p$ to $v$ we change from a visible to a non-visible region. By Lemma~\ref{lem_key}, this change occurs at the shadow of some visible reflex vertex. In this case,  $\textsc{NextVisReflex}(p,\C)$ should return this reflex vertex (as well as its shadow). For well-definedness purposes, we say that $\textsc{NextVisReflex}(p,\C)$ should return $b$ if $\C$ contains no reflex vertex.

The following observation (illustrated in Fig.~\ref{fig:ImprovedAlgorithm}) allows us to compute $\textsc{NextVisReflex}(p,\C)$ efficiently. 

\begin{figure}[tb]
  \centering
 \includegraphics{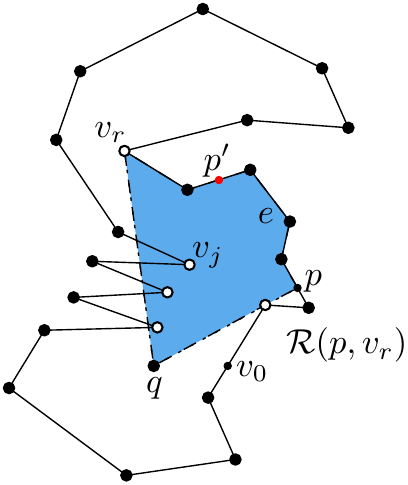}
\caption {\small When the first reflex vertex $v$ lying after $p$ is not visible, there must exist a visible reflex vertex in $\chain(p,b)$ angularly located between $p$ and $v$. The one with smallest CCW-angle inside $\region(p,v)$ ($v'$ in the figure) will determine the change between visible and non-visible regions between $p$ and $v$.}
\label{fig:ImprovedAlgorithm}
 \end{figure}



\begin{observation}\label{obs:SmallestAngle}
For any independent chain $\mathcal{C}=\chain(a,b)$, and visible point $p\in \C$ that is not an R-type reflex vertex, let $v$ be the first reflex vertex encountered when walking from $p$ towards $b$ on $\C$ (or $b$ if none exists). If $v$ is not visible, then the CCW-maximal chain starting at $p$ ends at the shadow of the L-type reflex vertex in $\region(p,v)$ with smallest CCW-angle.
\end{observation}

\begin{lemma}\label{lem_next}
For any independent chain $\C$ of $n$ vertices and a visible point $p\in \C$, \linebreak$\textsc{NextVisReflex}(p,\C)$ can be computed in $O(n)$ time using $O(1)$ additional variables.
\end{lemma}
\begin{proof}
Let $\mathcal{C}=\chain(a,b)$ and let $v$ be the first reflex vertex lying after $p$ on $\C$.
In $O(n)$ time we can perform a ray shooting query: if $v$ is visible then we output it (and its shadow) as the result of $\textsc{NextVisReflex}(p,\C)$. Otherwise, we use Observation~\ref{obs:SmallestAngle} and compute the reflex vertex on $\chain(p,b)$ with smallest CCW-angle (among those that are in $\region(p,v)$). This vertex is found by walking along $\chain(p,b)$ and keeping track of every time we enter or leave $\region(p,v)$. Note that since $p$ is visible and $\C$ is simple, we can only enter or leave $\region(p,v)$ when we cross line segment $qv$, hence this can be checked in constant time per edge of \C. Since a constant number of operations is needed per vertex, at most $O(n)$ time will be needed for computing $\textsc{NextVisReflex}(p,\C)$. 

Finally note that Observation~\ref{obs:SmallestAngle}  holds whenever $p$ is not an R-type reflex vertex. Thus, if $p$ is a reflex vertex of type R, it suffices to first compute its shadow, and  return the same as  $\textsc{NextVisReflex}(\textsc{Shadow}(p,\C),\C)$ would.
\end{proof}

Given an  independent chain $\C = (c_1, \ldots, c_n )$, our base algorithm works as follows: start from $c_1$, use $\textsc{NextVisReflex}$ to obtain the next visible reflex vertex, and report the CCW-maximal chain starting from $c_1$. We then repeat the procedure starting from the last reported vertex until we reach $c_n$ (see the details in Algorithm \ref{alg:subroutine}). 

\begin{theorem}\label{theo:Sequential algorithm}
Algorithm~\ref{alg:subroutine} reports the visibility polygon of an independent chain of $n$ vertices and $\overline{r}$ visible reflex vertices in counterclockwise order in $O(n\overline{r})$ time, using $O(1)$ additional variables.
\end{theorem}
\begin{proof}
Correctness of the algorithm is given by Lemma \ref{lem_next} and Corollary~\ref{corollary:CharacterizationVisibleReflex}.  It is easy to verify that Algorithm \ref{alg:subroutine} uses a constant number of variables, hence it remains to show a bound on the running time. Notice that at each iteration of the algorithm we report a visible reflex vertex. Hence, we can charge the cost of $\textsc{NextVisReflex}$ operation to the reported vertex. Since no vertex is reported twice and operation $\textsc{NextVisReflex}$ needs linear time, the total running time is bounded by $O(n\overline{r})$. 
\end{proof}

\begin{algorithm}
  \begin{algorithmic}[1]  
	\STATE $c_{\mathrm{start}} \leftarrow$ $c_1$ (or $c_{\mathrm{start}}\gets \textsc{Shadow}(c_1,\C)$ if $c_1$ is an R-type reflex vertex)
	\REPEAT
		\STATE $v\gets  \textsc{NextVisReflex}(c_{\mathrm{start}},\C)$
		\IF{$v=c_n$}
			\STATE  (*  The remainder of the chain is visible *)
			\STATE $c_{\mathrm{stop}} \gets c_n$
			\STATE $c_{\mathrm{next}}\gets c_n$
		\ELSE 
		        \STATE (* We found next visible reflex $v$. The reported chain will depend on the type of $v$ *)
			\IF{$v$ is of type R}
				\STATE $c_{\mathrm{stop}} \gets v$
				\STATE $c_{\mathrm{next}}\gets \textsc{Shadow}(v,\C)$
			\ELSE
				\STATE $c_{\mathrm{stop}} \gets \textsc{Shadow}(v,\C)$
				\STATE $c_{\mathrm{next}}\gets v$
			\ENDIF
		\ENDIF
    \STATE Report every vertex between $c_{\mathrm{start}}$ and $c_{\mathrm{stop}}$
    \STATE $c_{\mathrm{start}}\gets c_{\mathrm{next}}$
    \UNTIL{$c_{\mathrm{start}} = c_n$}
  \end{algorithmic}
\caption{Computing the visibility polygon of an independent chain $\C=(c_1,\ldots c_n)$}
\label{alg:subroutine}
\end{algorithm}

\section{A divide-and-conquer approach}\label{Section:Divide and conquer}

We now consider the case in which we are allowed a slightly larger amount of variables. We parametrize the running time of our algorithms by the number of working variables allowed, which we denote by $s$. Our aim is to obtain an algorithm whose running time decreases as $s$ grows. 

Using the result of the previous section as base algorithm, we now present a divide-and-conquer approach to solve the problem. The general scheme of our algorithm is the natural one: choose a reflex vertex $z$ inside the cone $\coneC$, perform a ray shooting query to find the visible point in the direction of $z$, and split the polygonal chain into two smaller independent subchains $\mathcal C_1, \mathcal C_2$ (see Fig.~\ref{fig:OutlineAlgorithm}, left). We repeat the process recursively, until either 
$(1)$ a chain $\mathcal C$ has a constant number of reflex vertices (see Fig.~\ref{fig:OutlineAlgorithm}, right) or 
$(2)$ the depth of the recursion is such that we would exceed the number of allowed working variables. 
Whenever either of these two conditions is met, we compute the visibility polygon of the chain using Algorithm \ref{alg:subroutine}. See a scheme of this approach in Algorithm~\ref{alg:DivideAndConquer}.


\begin{figure}[tb]
\centering
\includegraphics{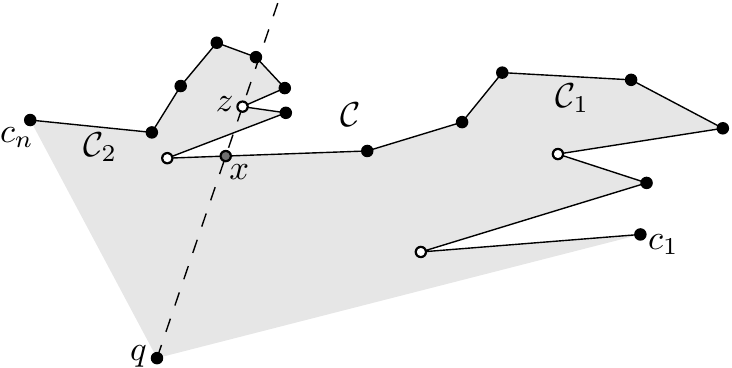}
\includegraphics{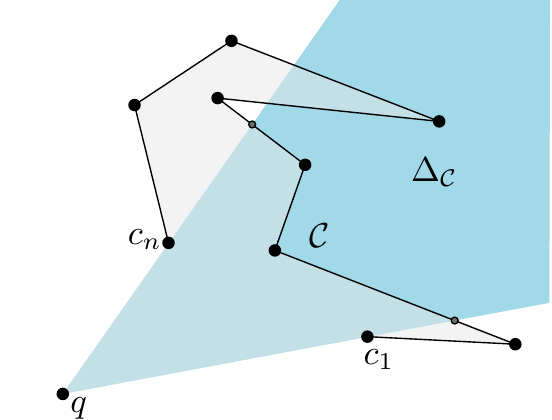}
\caption{Left: Split of $\mathcal C$ into two subchains $\mathcal C_1, \mathcal C_2$ using a visible point $x$ in the direction of a reflex vertex $z$.
Right: A polygonal chain $\mathcal C$ with no reflex vertices inside the cone $\coneC$, only one subchain of $\mathcal C$ is visible. }
\label{fig:OutlineAlgorithm}
\end{figure}

\begin{algorithm}
\caption{Given a polygonal chain $\mathcal C = (c_1, \ldots, c_n)$ such that $c_1, c_n$ are both visible points of $\mathcal P$ and a positive integer depth $d$ (initially $1$), compute $\VisC$}
\label{alg:DivideAndConquer}
\begin{algorithmic}[1]
\STATE $k\gets$ number of reflex vertices of $\mathcal C$ inside the cone $\coneC$
\IF{$k \leq 2$ or $d\geq h(s)$}
	\STATE\label{algoconst} Run Algorithm~\ref{alg:subroutine} on $\C$
\ELSE
	\STATE $v\leftarrow \textsc{findPartitionVertex}(\mathcal C)$
	\STATE $x\gets \textsc{RayShooting}(v,\C)$
	\STATE Algorithm~\ref{alg:DivideAndConquer}$ (\{c_1, \ldots, x\},d+1)$
	\STATE Algorithm~\ref{alg:DivideAndConquer}$ (\{x, \ldots, c_n\},d+1)$
\ENDIF
\end{algorithmic}
\end{algorithm}

In order to control the depth of the recursion we use a depth counter, hence the algorithm stops dividing once $d=h(s)$ (for some value $h(s)\in O(s)$ that will be determined later). Note that the split direction is decided by subroutine $\textsc{FindPartitionVertex}(\mathcal C)$. This procedure should give a direction so that the resulting subchains ${\mathcal C_1}$ and ${\mathcal C_2}$ have roughly the same complexity. 
 Naturally, it must also run reasonably fast and use $O(s)$ variables.

In this section we propose two different methods to choose the partition vertex. The first one uses the approximate median finding algorithm of  Raman and Ramnath~\cite{rr-iubtstsls-98}. The second one is randomized, and simply chooses a random reflex vertex among those lying in the cone $\coneC$. We first show that, regardless of the partition method used, the visibility polygon will be correctly computed. 

\begin{lemma}\label{lem_correctdiv}
Algorithm~\ref{alg:DivideAndConquer} correctly reports the visibility polygon of an independent chain in counterclockwise order, using $O(s)$ variables.
\end{lemma}
\begin{proof}
The divide-and-conquer approach repeatedly partitions the input into independent chains. Each subchain will eventually be reported. By Observation~\ref{DivideObs}, the union of reported vertices is equal to the visibility polygon, hence correctness is derived from the correctness of  Algorithm~\ref{alg:subroutine}.



Regarding space, the subroutines called in this algorithm (\textsc{findPartitionVertex} and \linebreak\textsc{RayShooting}) use $O(s)$ and $O(1)$ variables, respectively. Once the procedure finishes, their working space can be reused for further calls, hence we never use more than $O(s)$ working space at the same time. It remains to consider the memory used implicitly for handling the recursion. Since each step of the algorithm needs a constant number of variables, the total memory needed will be proportional to the recursion depth. Since $h(s)\in O(s)$, the claim is shown.
\end{proof}

In what follows we present two implementations of $\textsc{FindPartitionVertex}$. 

\subsection{Deterministic variant}\label{secdeter}

We start by giving a deterministic algorithm for $\textsc{FindPartitionVertex}(\mathcal C)$ using $O(s)$ extra variables. For this purpose, we will use the algorithms by Raman and Ramnath presented in~\cite{rr-iubtstsls-98}, which compute an approximation of the median value of a given set using reduced workspace.

Given a chain $\C$, let $\Theta = \{\theta_q(v_i) : v_i\text{ is a reflex vertex of $\C$ lying inside }\coneC\}$. An element $\theta_q(v_i)\in \Theta$ is called a \emph{$2/3$-median of $\Theta$} if there are at most $2|\Theta|/3$ elements in $\Theta$ smaller that $\theta_q(v_i)$ and at most $2|\Theta|/3$ greater that $\theta_q(v_i)$.
Given two elements $z,z'$ of $\Theta$ such that $z<z'$, we say that $z,z'$ is an \emph{approximate median pair} if at most $|\Theta|/2$ elements of $\Theta$ are smaller than $z$, at most $|\Theta|/2$ lie between $z$ and $z'$, and at most $|\Theta|/2$ elements are greater that $z'$.
Notice that if $z,z'$ is an approximate median pair, then either $z$ or $z'$ is a $2/3$-median of $\Theta$.
Moreover, we can determine which of the two is a $2/3$-median with one scan of $\C$. Thus, we say that every approximate median pair induces a $2/3$-median of $\Theta$.

Raman and Ramnath~\cite{rr-iubtstsls-98} presented two algorithms to find an approximate median pair. The first one is used whenever $s\in o(\log \Rin)$, and allows us to find a $2/3$-median of $\Theta$ in $O(sn\Rin^{1/s})$ time using $O(s)$ variables (Lemma 5 of~\cite{rr-iubtstsls-98}, assigning their parameter $p$ to $\Rin^{1/s}$). 
For the case $s\in \Omega(\log \Rin)$, they propose another algorithm (stated in Lemma 3 of~\cite{rr-iubtstsls-98}) that computes an approximate median pair in $O(n\log \Rin)$ time, using $O(\log \Rin)$ variables. We note that Raman and Ramnath used these algorithms to afterwards obtain the exact median, but in here we only need an approximation. 

$\textsc{FindPartitionVertex}(\mathcal C)$ will execute the approximate median pair algorithm, and return the reflex vertex $v$ that induces a $2/3$-median of $\Theta$. By construction, each of the two cones obtained from $\coneC$, by shooting a ray through $v$, will contain at most $2/3$ of the reflex vertices in $\coneC$. 
 
Let $P(n,\Rin)$ be the running time of the approximate median finding algorithm on a chain $\C$ of length $n$ with $\Rin$ reflex vertices inside $\coneC$ (that is, $P(n,r)=O(sn\Rin^{1/s})$ if $s\in o(\log \Rin)$, $P(n,r)=O(n\log \Rin)$ otherwise). We set $h(s)=s \log_{3/2}2 \approx 1.71 s$ (if $s\in o(\log r)$) or $s=\log_{3/2}r$ (otherwise). Observe that in both cases we have $h(s)\in O(s)$. We now prove an upper bound on the running time of Algorithm~\ref{alg:DivideAndConquer} with this implementation of $\textsc{FindPartitionVertex}(\mathcal C)$.

\begin{theorem}\label{theo_determ}
For any $s\in O(\log r)$, $\VisC$ can be computed in $O(\frac{nr}{2^{s}}+n\log^2 r)$  time using $O(s)$ variables.
\end{theorem}
\begin{proof}
Consider any node $u_i$ of the recursion tree of the algorithm and let $\mathcal C_i$ be the chain processed at this node. Let $n_i$ and $r_i$ be the size of $\mathcal C_i$ and the number of reflex vertices lying inside $\Delta_{\mathcal C_i}$, respectively.

If $u_i$ is a non-terminal node of the recursion, then the running time at $u_i$ is bounded by $O(P(n_i,r_i))$. 
Note that a vertex of the input can only appear in at most two chains of the same depth. Hence, the total cost of all non-terminal nodes of a fixed level $j$ in the recursion tree is bounded by $\sum_{u_i} O(P(n_i, r_i)) \leq O(P(n,r))$. By definition, there are at most $h(s)$ levels of recursion, hence the time spent in all the non-terminal executions of the algorithm is bounded by $O(P(n,r)h(s))$.

It remains to consider the time spent in the terminal nodes. Each terminal node $u_i$ will need $O(n_i\overline{r}_i)$ time, where $\overline{r}_i$ denotes the number of visible reflex vertices on $\mathcal C_i$. Recall that terminal nodes are only executed whenever either $\C_i$ has a constant number of reflex vertices or we have reached $h(s)$ levels of recursion. Further note that, at each level of recursion at least a third of the reflex vertices are discarded. In particular, we have that $\overline{r_i}\leq r(\frac{2}{3})^{h(s)}$. 

Similar to non-terminal nodes, vertices cannot be present in more than two terminal nodes. Thus, the total time spent in the terminal nodes is bounded by $$\sum_{u_i} O(n_i\overline{r_i}) \leq \sum_{u_i} O(n_ir(2/3)^{h(s)})) \leq O(nr(2/3)^{h(s)}).$$ Therefore, the total time spent by the algorithm becomes  $O(P(n,r)h(s)+nr(2/3)^{h(s)})$. 

This expression can be simplified by distinguishing between different workspace sizes. 

\begin{description}
\item[Case $s\in o(\log r)$.] In this case we have $P(n,r)h(s)=O(s^2nr^{1/s})$, and in particular $O(s^2nr^{1/s}) \in O(nr^{1/3})$ (since $s$ is a parameter that can be chosen to be at least 3). Further recall that $h(s)=s\log_{3/2}2$, and in particular  $(2/3)^{h(s)}=2^{-s}$. Thus, the second term simplifies to $\frac{nr}{2^s}$. 
Since $s\leq \log_2(r)/3$ and the running time decreases as $s$ grows, we have   $\frac{nr}{2^s} \geq \frac{nr}{r^{1/3}} \in \Omega(n\sqrt{r})$. 

That is, the running time of our algorithm is expressed as the sum of two terms that, when $s\in o(\log r)$, the first one is at most $O(nr^{1/3})$ whereas the second one is at least $\Omega(n\sqrt{r})$. Asymptotically speaking, the first one can be ignored, and the running time is dominated by $O(\frac{nr}{2^s})$ (i.e., applying Algorithm~\ref{alg:subroutine} to each terminal node). 
\item[Case $s\in\Theta(\log r)$.] In this case we can use the faster method of Raman and Ramnath (i.e. $P(n,r)=O(n\log r)$). Recall that in this case we set $h(s)=\log_{3/2}r$, hence  $(2/3)^{h(s)}=1/r$. In particular, the running time of the second term simplifies to $O(nr(2/3)^{h(s)})= 
O(n)$. Since  $s\in\Theta(\log r)$, the running time is dominated by the first term (i.e., finding the split direction), which is $O(P(n,r)h(s))=O(n\log^2 r)$. 
\end{description}
Observe that in both cases the running time is bounded by $O(\frac{nr}{2^{s}}+n\log^2 r)$, hence the claim holds. 
\end{proof}

{\bf Remark} Note that, although we only consider algorithms that use up to $O(\log r)$ variables, one could study what happens whenever more space is allowed. However, we note that increasing the size of our workspace will not reduce the running time, since the time bottleneck of this approach is determined by the approximate median method of Raman and Ramnath. 



\subsection{Randomized approach}\label{second}
Whenever $s\in \Theta (\log n)$, the running time of the previous algorithm is dominated by the \textsc{FindPartitionVertex} procedure. Motivated by this, in this section we consider a faster (albeit randomized) partition method. 

The randomized method proceeds as follows: let $k$ be the  number of reflex vertices of $\C$ lying inside $\coneC$ (note that $k$ can be computed in linear time by scanning $\C$). Select a random number $i$ between $1$ and $k$ (uniformly at random). The idea is to output the $i$-th reflex vertex inside $\coneC$ (computed by walking counterclockwise along $\C$). However, we must first check that this vertex will make a balanced partition. In order to check so, we make another scan of $\mathcal C$, and we count the angular rank of the $i$-th reflex vertex (among the reflex vertices in $\coneC$). If its rank is between $k/3$ and  $2k/3$, we use it as partition vertex. Otherwise, we pick another random number and repeat the process until such a vertex is found.

\begin{lemma}\label{lem_runningtime}
The randomized version of \textsc{FindPartitionVertex} has expected running time $O(n)$.
\end{lemma}
\begin{proof}
The probability that, when choosing a reflex vertex uniformly at random, we pick one whose rank is between $1/3$ and $2/3$ is exactly $1/3$. If each time we make the choices independently, we are performing Bernoulli trials whose probability of success is $1/3$. Hence, the expected number of times we have to choose a random index is a constant (three in this case). For each try we only need to check its rank (which can be done in $O(n)$ time by performing two scans of the input).
Therefore we conclude that the expected running time of \textsc{FindPartitionVertex} is $O(n)$.
\end{proof}

\begin{theorem}\label{theo_randomi}
For any $s\in O(\log r)$, $\VisC$  can be computed in $O(\frac{nr}{2^{s}}+n\log r)$ expected time using $O(s)$ variables.
\end{theorem}
\begin{proof}
The proof is identical to the proof of Theorem \ref{theo_determ}, just taking into account that now $P(n,r)=O(n)$, hence the running time of the second term is decreased by a $\log r$ factor. 
\end{proof}

Observe that this algorithm is faster than the previous one when $s\in \Theta(\log r)$.
We conclude by summarizing the different algorithms presented in this paper. 

\begin{theorem}
Given a polygon $\Poly$ of $n$ vertices, out of which $r$ are reflex, the visibility polygon of a point $q\in\Poly$ can be computed in $O(n\Rout)$ time using $O(1)$ variables (where $\Rout$ is the number of reflex vertices of $\Vis$) or $O(\frac{nr}{2^{s}})$ time using $O(s)$ variables (for any $s\in o(\log r)$). If $\Omega(\log r)$ variables are available, the running time decreases to $O(n\log^2 r)$ time or $O(n\log r)$ randomized expected time.
\end{theorem}

\bibliographystyle{abbrv}
\bibliography{visi}

\end{document}